\documentclass[a4paper]{llncs}

\usepackage{latexsym}
\usepackage{amssymb,amsmath}
\usepackage{url}
\usepackage{hyperref}
\usepackage{pifont} 

\usepackage{enumerate}

\usepackage{tikz}
\usetikzlibrary{automata}
\usetikzlibrary{positioning,arrows,fit,calc}
\usetikzlibrary{decorations.pathmorphing}
\usetikzlibrary{shapes}
\usetikzlibrary{backgrounds}
\usepackage{multirow}
\usepackage{booktabs}
\usepackage{xcolor,colortbl}
\usepackage{setspace}

\usepackage{hhline}
\usepackage{cite}

\newcommand{\Tmabc}{\Gamma}
\newcommand{\Aiabc}{\Sigma}
\newcommand{\Tmtran}{\gamma}
\newcommand{\Aabc}{\Sigma_+}
\newcommand{\pp}{p}
\newcommand{\ppn}{N}
\newcommand{\tst}{\textit{tr}}

\newcommand{\perm}[3]{\fu^{{#1}_{#2}^{#3}}}

\newcommand{\yes}{$\mathsf{yes}$}
\newcommand{\no}{$\mathsf{no}$}
\newcommand{\RPR}{\mathsf{RPR}}

\newcommand{\A}{\ensuremath{\mathfrak A}}

\newcommand{\fu}{\delta}

\newcommand{\lang}{\mathcal{L}}

\renewcommand{\L}{{\boldsymbol{L}}}

\newcommand{\Qr}{Q^r}
\newcommand{\simclass}[1]{#1/_{\mathop{\simm}}}

\newcommand{\LTL}{\textsl{LTL}}

\newcommand{\NCo}{{{\ensuremath{\textsc{NC}^1}}}}
\newcommand{\ACz}{{\ensuremath{\textsc{AC}^0}}}
\newcommand{\ACC}{{\ensuremath{\textsc{ACC}^0}}}

\newcommand{\MOD}{\ensuremath{\mathsf{MOD}}}
\newcommand{\coNP}{{\ensuremath{\textsc{coNP}}}}
\newcommand{\ExpSpace}{{\ensuremath{\textsc{ExpSpace}}}}
\newcommand{\PSpace}{{\ensuremath{\textsc{PSpace}}}}

\newcommand{\FO}{\mathsf{FO}}

\newcommand{\qa}{q_{\textit{acc}}}
\newcommand{\B}{\mathsf{b}}
\newcommand{\M}{\boldsymbol{M}}
\newcommand{\conf}{\mathfrak c}

\newcommand{\aut}{\mathfrak B}
\newcommand{\G}{\mathfrak G}
\newcommand{\Sg}{\mathfrak S}
\newcommand{\unit}{identity}
\newcommand{\fum}{\tilde{\fu}}
\newcommand{\ord}{o}
\newcommand{\id}{\mathsf{id}}

\newcommand{\simm}{\sim}
\newcommand{\Amin}{\A_{\L}}

\tikzset{
  basic box/.style = {
    shape = rectangle,
    align = center,
    draw  = #1,
    rounded corners},
  header node/.style = {
    font          = \strut\Large\ttfamily,
    text depth    = +0pt,
    fill          = white,
    draw},
  header/.style = {%
    inner ysep = +1.5em,
    append after command = {
      \pgfextra{\let\TikZlastnode\tikzlastnode}
      node [header node] (header-\TikZlastnode) at (\TikZlastnode.north) {#1}
    }
  },
  hv/.style = {to path = {-|(\tikztotarget)\tikztonodes}},
  vh/.style = {to path = {|-(\tikztotarget)\tikztonodes}},
  fat blue line/.style = {ultra thick, blue}
}

\title{Deciding FO-definability of Regular Languages}

\author{Agi Kurucz\inst{1} \and Vladislav Ryzhikov\inst{2} \and Yury Savateev\inst{2,3} \and Michael Zakharyaschev\inst{2,3}}

\institute{King's College London, UK \and Birkbeck, University of London, UK \and HSE University, Moscow, Russia}

\begin{document}

\maketitle


\begin{abstract}
We prove that, similarly to known \PSpace-completeness of recognising $\FO(<)$-definability of the language $\L(\A)$ of a DFA $\A$,
deciding both \mbox{$\FO(<,\equiv)$-} and $\FO(<,\mathsf{MOD})$-de\-fin\-abil\-ity 
(corresponding to circuit complexity in $\ACz$ and $\ACC$) are \PSpace-complete.
We obtain these results by first showing that known algebraic characterisations of FO-definability
of $\L(\A)$ can be captured by `localisable' properties of the transition monoid of $\A$.
Using our criterion, we then generalise the known proof of \PSpace-hardness of $\FO(<)$-definability, and establish the upper bounds not only for arbitrary DFAs but also for 2NFAs. 
\end{abstract}

\section{Introduction}

This paper gives answers to some open 
questions related to finite automata, logic and circuit complexity.
Research in this area goes back (at least) to the early 1960s when B\"{u}chi~\cite{Buchi60}, Elgot~\cite{Elgot61} and Trakhtenbrot~\cite{Trakh62} showed that \mbox{$\text{MSO}(<)$} (monadic second-order) sentences over finite strict linear orders define exactly the class of regular languages. 

$\FO(<)$-definable regular languages were proven to be the same as star-free languages~\cite{McNaughton&Papert71}, and their algebraic characterisation as languages with aperiodic syntactic monoids was obtained in~\cite{DBLP:journals/iandc/Schutzenberger65a}. Algebraic characterisations of FO-definability in other signatures, and
circuit and descriptive complexity of regular languages were investigated in~\cite{DBLP:journals/jcss/Barrington89,DBLP:journals/jcss/BarringtonCST92,Straubing94}, which established an $\ACz$/$\ACC$/$\NCo$ trichotomy.
In particular,
the regular languages decidable in \ACz{} are definable by $\FO(<,\equiv)$-sentences with unary predicates $x \equiv 0\, (\text{mod}\ n)$; those in $\ACC$ are definable by $\FO(<,\MOD)$-sentences with quantifiers $\exists^n x\, \psi(x)$ checking whether the number of positions satisfying $\psi$ is divisible by $n$; and all regular languages are definable in $\FO(\RPR)$ with relational primitive recursion~\cite{DBLP:journals/iandc/ComptonL90}; see Table~\ref{tab:algebra}. 

The problem of deciding whether the language of a given DFA $\A$ 
is $\FO(<)$-definable is known to be \PSpace-complete~\cite{DBLP:journals/iandc/Stern85,DBLP:journals/TCS/ChoHyunh91,DBLP:journals/actaC/Bernatsky97} (which is also a special case of general results on finite 
monoids~\cite{Beaudryetal92,fleischeretal18}).
As shown in \cite{DBLP:journals/jcss/BarringtonCST92}, the algebraic criteria of Table~\ref{tab:algebra} yield 
algorithms deciding whether a given regular language is in $\ACz$ and $\FO(<,\equiv)$-definable, or in $\ACC$ and $\FO(<,\MOD)$-definable, or $\NCo$-complete and is not $\FO(<,\MOD)$-definable (unless $\ACC = \NCo$). 
However, these `brute force' algorithms are not optimal, requiring the generation of the whole transition monoid 
of $\A$, which can be of exponential size~\cite{DBLP:journals/eatcs/HolzerK04}. 
As far as we know, the precise complexity of these decision problems has remained open. 

\begin{table}[ht]
\centering
\begin{tabular}{c|c|c}\toprule
definability of $\L$ & algebraic characterisation of $\L$ & circuit complexity\\
\hline
$\FO(<)$& $M(\L)$ is aperiodic & \multirow{2}{*}{in \ACz}\\
\hhline{|-|-|~|}
$\FO(<,\equiv)$& $\eta_\L$ is quasi-aperiodic &  \\
\hline
$\FO(<,\MOD)$ &all groups in $M(\L)$ are solvable &in \ACC\\
\hline
$\FO(\RPR)$ &arbitrary $M(\L)$& in \NCo\\
\hline
\hline
not in $\FO(<,\MOD)$  & $M(\L)$ contains an unsolvable group & \NCo-hard\\
\bottomrule
\end{tabular}
\caption{Definability, algebraic characterisations and circuit complexity of a regular language $\L$, where $M(\L)$ is the syntactic monoid and $\eta_\L$ the syntactic morphism of $\L$.}
\label{tab:algebra}
\end{table}

Our interest in the exact complexity of these problems is motivated by recent advances in ontology-based data access (OBDA) with linear time temporal logic \LTL~\cite{DBLP:conf/time/ArtaleKKRWZ17,DBLP:journals/corr/abs-2004-07221}.
The classical (atemporal) OBDA paradigm~\cite{PLCD*08,DBLP:conf/ijcai/XiaoCKLPRZ18} relies on a reduction of answering a query mediated by an ontology under the open-world semantics to evaluating a database query in a standard language such as SQL or its extension---that is, essentially, an extension of first-order logic---under the closed-world semantic. In the context of temporal OBDA, answering \LTL{} ontology-mediated queries is equivalent to deciding certain regular languages given by an NFA or 2NFA of (possibly) exponential size, which gives rise to the circuit complexity and FO-definability problems for those languages.
For further details the reader is referred to \cite{DBLP:conf/time/21}, which relies on the results we obtain below. 

\smallskip
\noindent
\emph{Our contribution} in this paper is as follows.
Let $\lang$ be one of the languages 
$\FO(<,\equiv)$ or $\FO(<,\MOD)$. 
First, using the algebraic characterisation results of ~\cite{DBLP:journals/jcss/Barrington89,DBLP:journals/jcss/BarringtonCST92,Straubing94}, we
obtain criteria for the $\lang$-definability of the language $\L(\A)$ of any given DFA $\A$ in terms of a limited part of the transition monoid of $\A$ (Theorem~\ref{DFAcrit}). Then,
by using our criteria and generalising the construction of~\cite{DBLP:journals/TCS/ChoHyunh91},  we show that deciding $\lang$-definability of $\L(\A)$ for any minimal DFA $\A$ is \PSpace-hard (Theorem~\ref{DFAhard}).
Finally, we apply our criteria to give a \PSpace-algorithm deciding $\lang$-definability of $\L(\A)$ for not only any DFA but any 2NFA $\A$  (Theorem~\ref{thm:2NFA}).


\section{Preliminaries}\label{prelims:aut}

We begin by briefly reminding the reader of the basic algebraic and automata-theoretic notions required in the remainder of the paper.

\subsection{Monoids and Groups}

A \emph{semigroup} is a structure $\Sg=(S,\cdot)$ where $\cdot$ is an associative binary operation.
Given $s,s'\in S$ and $n>0$, we write $s^n$ for $s\cdot$ $\dots$ $\cdot s$ $n$-times, and often write $ss'$ for $s\cdot s'$.
An element $s$ in a semigroup $\Sg$ is \emph{idempotent} if $s^2=s$.
An element $e$ in $\Sg$ is an \emph{\unit{}} if $e\cdot x=x\cdot e=x$ for all $x\in S$.
(It is easy to see that such an $e$ is unique, if exists.) The \unit{} element is clearly idempotent.
A \emph{monoid} is a semigroup with an \unit{} element.
For any element $s$ in a monoid, we set $s^0=e$.
A monoid $\Sg=(S,\cdot)$ is a \emph{group} if, for any $x\in S$, there is $x^-\in S$---the \emph{inverse of} $x$---such that
$x\cdot x^-=x^-\cdot x=e$ (every element of a group has a unique inverse). A group is \emph{trivial} if it has one element, and \emph{nontrivial} otherwise.

Given two groups $\G=(G,\cdot)$ and $\G'=(G',\cdot')$,
a map $h\colon G\to G'$ is a \emph{group homomorphism from $\G$ to $\G'$} if  $h(g_1\cdot g_2)=h(g_1)\cdot' h(g_2)$ for all
$g_1,g_2\in G$. (It is easy to see that any group homomorphism maps the \unit{} of $\G$
to the \unit{} of $\G'$ and preserves the inverses. The set
$\{h(g)\mid g\in G\}$
is closed under $\cdot'$, and so is a group, the \emph{image of $\G$ under $h$\/}.)
$\G$ is a \emph{subgroup of} $\G'$ if $G\subseteq G'$ and
the identity map $\id_G$ is a group homomorphism.
Given $X\subseteq G$, the \emph{subgroup of $\G$ generated by $X$} is the smallest subgroup of $\G$ containing $X$.
%
%
The \emph{order} $\ord_\G(g)$ of an element $g$ in $\G$ is
the smallest positive number $n$ with $g^n=e$, which always exists. Clearly, $\ord_\G(g)=\ord_\G(g^-)$ and, if $g^k=e$
then $\ord_\G(g)$ divides $k$. Also,
\begin{equation}\label{idemp}
\mbox{if $g$ is a nonidentity element in a group $\G$, then $g^k\ne g^{k+1}$ for any $k$.}
\end{equation}
A semigroup $\Sg'=(S',\cdot')$ is a \emph{subsemigroup} of a semigroup $\Sg=(S,\cdot)$
if $S'\subseteq S$ and $\cdot'$ is the restriction of $\cdot$ to $S'$.
Given a monoid $\M=(M,\cdot)$ and a set $S\subseteq M$, we say that $S$ \emph{contains the group} $\G=(G,\cdot')$, if
$G\subseteq S$ and $\G$ is a subsemigroup of $\M$.
Note that we do {\bf not} require the \unit{} of $\M$ to be in $\G$, even if it is in $S$.
If $S=M$, we also say that $\M$ \emph{contains the group} $\G$, or $\G$ \emph{is in} $\M$. We call a monoid $\M$ \emph{aperiodic} if it does not contain any nontrivial groups.

Let $\Sg=(S,\cdot)$ be a finite semigroup and $s\in S$. By the pigeonhole principle,
there exist $i,j\geq 1$ such that $i+j\leq |S|+1$ and $s^{i}=s^{i+j}$. Take the minimal such numbers, that is,
let $i_s,j_s\geq 1$ be such that $i_s+j_s\leq |S|+1$  and $s^{i_s}=s^{i_s+j_s}$ but $s^{i_s},s^{i_s+1},\dots,s^{i_s+j_s-1}$ are all different.
Then clearly $\G_s=(G_s,\cdot)$, where $G_s=\{s^{i_s},s^{i_s+1},\dots,s^{i_s+j_s-1}\}$, is a subsemigroup of $\Sg$.
It is easy to see that there is $m\geq 1$ with $i_s\leq m\cdot j_s<i_s+j_s\leq |S|+1$, and so $s^{m\cdot j_s}$ is idempotent. Thus, for every element $s$ in a semigroup $\Sg$, we have the following:
\begin{align}
\label{gini}
& \mbox{there is $n\geq 1$ such that $s^n$ is idempotent;}\\
\label{ginii}
& \mbox{$\G_s$ is a group in $\Sg$ (isomorphic to the cyclic group $\mathbb Z_{j_s}$);}\\
\label{giniii}
& \mbox{$\G_s$ is nontrivial iff $s^n\ne s^{n+1}$ for any $n$.}
\end{align}
%
%
%
 %
%
%
%
%
%
%
Let $\fu \colon Q\to Q$ be a function on a finite set $Q \ne \emptyset$.
For any $p\in Q$, the subset $\{\fu^k(p)\mid k<\omega\}$ with the obvious multiplication is a semigroup,
and so we have:
%
\begin{align}
\label{fpi}
& \mbox{for every $p\in Q$, there is $n_p\geq 1$ such that $\fu^{n_p}\bigl(\fu^{n_p}(p)\bigr)=\fu^{n_p}(p)$;}\\
\label{fpii}
& \mbox{there exist $q\in Q$ and $n\geq 1$ such that $q=\fu^n(q)$;}\\
\nonumber
& \mbox{for every $q\in Q$, if $q=\fu^k(q)$ for some $k\geq 1$,}\\
\label{fpiii}
& \hspace*{3.5cm}\mbox{then there is $n$, $1\leq n\leq |Q|$, with $q=\fu^n(q)$.}
\end{align}

For a definition of \emph{solvable} and \emph{unsolvable} groups the reader is referred to~\cite{rotman1999introduction}. Here, we only need the fact that any homomorphic image of a solvable group is solvable and the Kaplan--Levy criterion \cite{kaplan_levy_2010} (generalising Thompson's \cite[Cor.3]{thompson1968}) according to which a finite group $\G$ is unsolvable iff it contains three elements $a,b,c,$ such that $\ord_\G(a)=2$, $\ord_\G(b)$ is an odd prime,
$\ord_\G(c)>1$ and coprime to both $2$ and $\ord_\G(b)$, and $abc$ is the \unit{} element of $\G$.

A one-to-one and onto function on a finite set $S$ is called a \emph{permutation on} $S$.
The \emph{order of a permutation} $\fu$ is its order in the group of all permutations on $S$
(whose operation is composition, and its identity element is the identity permutation $\id_S$).
We use the standard cycle notation for permutations.

Suppose $\G$ is a monoid of $Q\to Q$ functions, for some finite set $Q \ne \emptyset$.
Let $S=\{q\in Q\mid e_\G(q)=q\}$, where $e_\G$ the \unit{} element in $\G$. For every function $\fu$ in $\G$, let $\fu\!\!\restriction_S$ denote
the restriction of $\fu$ to $S$.
Then we have the following:
\begin{align}
\label{groupini}
& \mbox{$\G$ is a group iff  $\fu\!\!\restriction_S$ is a permutation on $S$, for every $\fu$ in $\G$;}\\
\nonumber
& \mbox{if $\G$ is a group and $\fu$ is a nonindentity element in it, then $\fu\!\!\restriction_S\ne\id_S$ and}\\
\label{groupinii}
& \hspace*{3cm}\mbox{the order of the permutation $\fu\!\!\restriction_S$ divides $\ord_\G(\fu)$.}
\end{align}


\subsection{Automata: DFAs, NFAs, 2NFAs}

A \emph{two-way nondeterministic finite automaton} is a quintuple $\A = (Q, \Sigma, \delta, Q_0, F)$ that consists of an alphabet $\Sigma$, a finite set $Q$ of states with a subset $Q_0 \ne \emptyset$ of initial states and a subset $F$ of accepting states, and a transition function $\delta \colon Q \times \Sigma \to 2^{Q \times \{-1,0,1\}}$ indicating the next state and whether the head should move left ($-1$), right ($1$), or stay put. If $Q_0 = \{q_0\}$ and $|\delta(q, a)| = 1$, for all $q \in Q$ and $a \in \Sigma$, then $\A$ is \emph{deterministic}, in which case we write $\A = (Q, \Sigma, \delta, q_0, F)$.
If $\delta(q, a) \subseteq Q \times \{1\}$, for all $q \in Q$ and $a \in \Sigma$, then $\A$ is a \emph{one-way} automaton, and we write $\delta \colon Q \times \Sigma \to 2^Q$. As usual, DFA and NFA refer to one-way deterministic and non-deterministic finite automata, respectively, while 2DFA and 2NFA to the corresponding two-way automata. Given a 2NFA $\A$, we write $q \to_{a,d} q'$ if $(q', d) \in \delta(q,a)$; given an NFA $\A$, we write $q \to_{a} q'$ if $q' \in \delta(q,a)$.
A \emph{run} of a 2NFA $\A$ is a word in $(Q \times \mathbb N)^*$. A run $(q_0, i_0), \dots, (q_m, i_m)$ is a \emph{run of $\A$ on a word} $w = a_0 \dots a_n \in \Sigma^*$ if $q_0 \in Q_0$, $i_0 = 0$ and there exist $d_0, \dots, d_{m-1} \in \{-1,0,1\}$ such that $q_j \to_{a_j, d_j} q_{j+1}$ and $i_{j+1} = i_j+d_j$ for all $j$, $0 \leq j < m$. The  run is \emph{accepting} if $q_m \in F$, $i_m = n+1$. $\A$ \emph{accepts} $w \in \Sigma^*$ if there is an accepting run of $\A$ on $w$; the language $\L(\A)$ of $\A$ is the set of all words accepted by $\A$.

Given an NFA $\A$, states $q,q' \in Q$, and $w = a_0 \dots a_n \in \Sigma^*$, we write $q \to_w q'$ if either $w = \varepsilon$ and $q' = q$ or there is a run of $\A$ on $w$ that starts with $(q_0, 0)$ and ends with $(q', n+1)$. We say that a state $q \in Q$ is \emph{reachable} if $q' \to_w q$, for some $q' \in Q_0$ and $w \in \Sigma^*$.


Given a DFA $\A = (Q, \Sigma, \fu, q_0, F)$ and a
word $w \in \Sigma^\ast$, we define a function $\fu_w \colon Q \to Q$ by taking  $\fu_w(q) = q'$ iff $q \to_w q'$. We also define an equivalence relation $\simm$ on the set $\Qr\subseteq Q$ of reachable states by taking $q\simm q'$ iff,
for every $w \in \Sigma^\ast$, we have $\fu_w(q)\in F$ just in case $\fu_w(q')\in F$. We denote the $\simm$-class of $q$ by $\simclass{q}$, and let
$\simclass{X}=\{\simclass{q}\mid q\in X\}$ for any $X\subseteq \Qr$. Define $\fum_w\colon \simclass{\Qr\!}\to \simclass{\Qr\!}$ by taking
$\fum_w(\simclass{q})=\simclass{\fu_w(q)}$.
Then $\bigl(\simclass{\Qr\!},\Sigma,\fum,\simclass{q_0},\simclass{(F\cap \Qr)}\bigr)$ is the \emph{minimal DFA} whose language coincides with the language of $\A$.
Given a regular language $\L$, we denote by $\Amin$ the minimal DFA whose language is $\L$.

The \emph{transition monoid of} a DFA $\A$ is $M(\A) = (\{ \fu_w  \mid w \in \Sigma^\ast\},\cdot)$ with
$\fu_v\cdot\fu_w = \fu_{vw}$, for any $v,w$.
The \emph{syntactic monoid $M(\L)$ of $\L$} is the transition monoid $M(\Amin)$ of $\Amin$.
%
%
The \emph{syntactic morphism of} $\L$ is the map $\eta_\L$ from $\Sigma^*$ to the domain of $M(\L)$ defined by $\eta_\L(w) = \fum_w$. We call $\eta_\L$ \emph{quasi-aperiodic} if $\eta_L(\Sigma^t)$ is aperiodic for every $t<\omega$.

Suppose $\lang \in \{ \FO(<), \FO(<,\equiv), \FO(<,\MOD)\}$. A language $\L$ over an alphabet $\Sigma$ is \emph{$\lang$-definable} if there is an $\lang$-sentence $\varphi$ in the signature $\Sigma$, whose symbols are treated as unary predicates, such that, for any $w \in \Sigma^*$, we have $w=a_0\ldots a_n \in \L$ iff $\mathfrak S_w \models \varphi$, where $\mathfrak S_w$ is an FO-structure with domain $\{0,\dots,n\}$ ordered by $<$, in which $\mathfrak S_w \models a(i)$ iff $a=a_i$, for $0\leq i\leq n$.

%





Table~\ref{tab:algebra} summarises the known results that connect definability of a regular language $\L$ with properties of the syntactic monoid $M(\L)$ and syntactic morphism $\eta_\L$ (see~\cite{DBLP:journals/jcss/BarringtonCST92} for details) and with its circuit complexity under a reasonable binary encoding of $\L$'s alphabet (see, e.g.,~\cite[Lemma~2.1]{DBLP:journals/actaC/Bernatsky97}) and the assumption that $\ACC \ne \NCo$. We also remind the reader that a regular language is $\FO(<)$-definable iff it is star-free~\cite{Straubing94}, and that $\ACz \subsetneqq \ACC \subseteq \NCo$~\cite{Straubing94,DBLP:books/daglib/0028687}.


\section{Criteria of $\lang$-definability}

In this section, we show that the algebraic characterisations of FO-definability of $\L(\A)$ given in Table~\ref{tab:algebra} 
can be captured by `localisable' properties 
of the transition monoid of $\A$, for any given DFA $\A$.
Note that Theorem~\ref{DFAcrit}~$(i)$ was already observed in \cite{DBLP:journals/iandc/Stern85} and used in
proving that $\FO(<)$-definability of $\L(\A)$ is \PSpace-complete \cite{DBLP:journals/iandc/Stern85,DBLP:journals/TCS/ChoHyunh91,DBLP:journals/actaC/Bernatsky97}; while criteria $(ii)$ and $(iii)$ seem to be new. 

%
%

\begin{theorem}\label{DFAcrit}
For any DFA $\A=(Q,\Sigma,\delta,q_0,F)$, the following criteria hold\textup{:}
\begin{description}
\item[$(i)$] $\L(\A)$ is not $\FO(<)$-definable iff  $\A$ contains a nontrivial cycle, that is, there exist a word $u\in\Sigma^\ast$, a
state $q\in \Qr$, and a number $k\leq|Q|$ such that $q\not\simm \fu_u(q)$ and $q= \fu_{u^k}(q)$\textup{;}

\item[$(ii)$] $\L(\A)$ is not $\FO(<,\equiv)$-definable iff there are words $u,v\in\Sigma^\ast$, a
state $q\in \Qr$, and a number $k\leq |Q|$ such that $q\not\simm\fu_u(q)$, $q=\fu_{u^k}(q)$, $|v|=|u|$, and
$\fu_{u^i}(q)=\fu_{u^iv}(q)$, for every $i<k$\textup{;}


\item[$(iii)$] $\L(\A)$ is not $\FO(<,\MOD)$-definable iff there exist words $u,v\in\Sigma^\ast$, a
state $q\in \Qr$ and numbers $k,l\leq |Q|$ such that $k$ is an odd prime, $l>1$ and coprime to both $2$ and $k$,
$q\not\simm\fu_u(q)$, $q\not\simm\fu_v(q)$, $q\not\simm\fu_{uv}(q)$ and, for all $x\in\{u,v\}^\ast$, we have
$\fu_{x}(q)\simm\fu_{xu^{2}}(q)\simm\fu_{xv^{k}}(q)\simm\fu_{x(uv)^{l}}(q)$.
\end{description}
\end{theorem}


\begin{proof}
Throughout, we use the algebraic criteria of Table~\ref{tab:algebra} for $\L=\L(\A)$.
Thus,  $M(\L)$ is the transition monoid of the minimal DFA $\A_{\L(\A)}$, whose transition function we denote by $\fum$. 

$(i)~(\Rightarrow)$ Suppose $\G$ is a nontrivial group in $M(\A_{\L(\A)})$.
Let $u\in\Sigma^\ast$ be such that $\fum_u$ is a nonidentity element in $\G$.
We claim that there is $p\in \Qr$ such that $\fum_{u^n}(\simclass{p})\ne\fum_{u^{n+1}}(\simclass{p})$ for any $n>0$.
Indeed, otherwise for every $p\in \Qr$ there is $n_p>0$ with $\fum_{u^{n_p}}(\simclass{p})=\fum_{u^{n_p+1}}(\simclass{p})$.
Let $n=\max\{n_p\mid p\in \Qr\}$. Then $\fum_{u^n}=\fum_{u^{n+1}}$, contrary to~\eqref{idemp}.

By \eqref{fpi}, there is $m\geq 1$ with $\fum_{u^{2m}}(\simclass{p})=\fum_{u^{m}}(\simclass{p})$.
Let $\simclass{s}=\fum_{u^m}(\simclass{p})$. Then $\simclass{s}=\fum_{u^m}(\simclass{s})$, and so the restriction of $\fu_{u^m}$ to the subset $\simclass{s}$
of $\Qr$  is an $\simclass{s}\to \simclass{s}$ function.
By \eqref{fpii},
there exist $q\in \simclass{s}$ and $n\geq 1$
such that
$(\fu_{u^{m}})^n(q)=q$. Thus, $\fu_{u^{mn}}(q)=q$, and so by \eqref{fpiii}, there is
$k\leq |Q|$ with $\fu_{u^{k}}(q)=q$.
As $\simclass{s}\ne\fum_{u}(\simclass{s})$, we also have $q\not\simm\fu_u(q)$, as required.

$(i)~(\Leftarrow)$
Suppose the condition holds for $\A$.
Then there are $u\in\Sigma^\ast$,
$q\in \simclass{\Qr\!}$, and $k<\omega$ such that $q\ne\fum_u(q)$ and $q=\fum_{u^k}(q)$.
Then $\fum_{u^n}\ne\fum_{u^{n+1}}$ for any $n>0$. Indeed, otherwise we would have some $n>0$ with $\fum_{u^n}(q)=\fum_{u^{n+1}}(q)$. Let $i,j$ be such that $n=i\cdot k+j$ and $j<k$. Then
\[
q=\fum_{u^k}(q)=\fum_{u^{(i+1)k}}(q)=\fum_{u^n u^{k-j}}(q)=\fum_{u^{n+1}u^{k-j}}(q)=\fum_{u^{(i+1)k}u}(q)=\fum_u(q).
\]
So, by \eqref{ginii} and \eqref{giniii}, $\G_{\fum_u}$ is a nontrivial group in $M(\A_{\L(\A)})$.


$(ii)~(\Rightarrow)$
Let $\G$ be a nontrivial group in $\eta_\L(\Sigma^t)$, for some $t<\omega$, and
let $u\in\Sigma^t$ be such that $\fum_u$ is a nonidentity element in $\G$.
As shown in the proof of $(i)~(\Rightarrow)$,
there exist $s\in \Qr$ and $m\geq 1$ such that $\simclass{s}\ne\fum_{u}(\simclass{s})$ and
$\simclass{s}=\fum_{u^m}(\simclass{s})$.
Now let $v\in\Sigma^t$ be such that $\fum_v$ is the \unit{} element in $\G$, and consider $\fu_v$.
By \eqref{gini}, there is $\ell\geq 1$ such that $\fu_{v^\ell}$ is idempotent.
Then $\fu_{v^{2\ell-1}v^{2\ell}}=\fu_{v^{2\ell-1}}$. Thus, if we let $\bar{u}=uv^{2\ell-1}$ and $\bar{v}=v^{2\ell}$, then
$|\bar{u}|=|\bar{v}|$ and
$\fu_{\bar{u}^i}=\fu_{\bar{u}^i\bar{v}}$ for any $i<\omega$.
Also,  $\fum_{u^i}=\fum_{\bar{u}^i}$ for every $i\geq 1$, and so
the restriction of $\fu_{\bar{u}^m}$ to $\simclass{s}$ is an $\simclass{s}\to \simclass{s}$ function.
By \eqref{fpii},
there exist $q\in \simclass{s}$ and $n\geq 1$ such that
$(\fu_{\bar{u}^{m}})^n(q)=q$.
Thus, $\fu_{\bar{u}^{mn}}(q)=q$, and so by \eqref{fpiii}, there is some $k\leq |Q|$ with $\fu_{\bar{u}^{k}}(q)=q$.
As $\simclass{s}\ne\fum_u(\simclass{s})=\fum_{\bar{u}}(\simclass{s})$, we also have $q\not\simm\fu_{\bar{u}}(q)$, as required.

$(ii)~(\Leftarrow)$
If the condition holds for $\A$, then  there exist
$u,v\in\Sigma^\ast$, $q\in \simclass{\Qr\!}$, and $k<\omega$ such that $q\ne\fum_u(q)$, $q=\fum_{u^k}(q)$, $|v|=|u|$,
and $\fum_{u^i}(q)=\fum_{u^iv}(q)$, for every $i<k$.
As $M(\A_{\L(\A)})$ is finite, it has finitely many subsets.
So there exist $i,j\geq 1$ such that $\eta_\L(\Sigma^{i|u|})=\eta_\L(\Sigma^{(i+j)|u|})$.
Let $z$ be a multiple of $j$ with $i\leq z<i+j$.
Then $\eta_\L(\Sigma^{z|u|})=\eta_\L(\Sigma^{(z|u|)^2})$, and so $\eta_\L(\Sigma^{z|u|})$ is closed under the composition of functions (that is, the semigroup operation of $M(\A_{\L(\A)})$). Let $w=uv^{z-1}$ and consider the group
$\G_{\fum_w}$ (defined above \eqref{gini}--\eqref{giniii}). Then $G_{\fum_w}\subseteq\eta_\L(\Sigma^{z|u|})$.
We claim that $\G_{\fum_w}$ is nontrivial. Indeed, we have $\fum_w(q)=\fum_{uv^{z-1}}(q)=\fum_u(q)\ne q$.
On the other hand, $\fum_{w^{k}}(q)=\fum_{u^{k}}(q)=q$.
By the proof of $(i)~(\Leftarrow)$, $\G_{\fum_w}$ is nontrivial.

$(iii)~(\Rightarrow)$
Suppose $\G$ is an unsolvable group in $M(\A_{\L(\A)})$. By the Kaplan--Levy criterion,
$\G$ contains three functions $a,b,c$ such that $\ord_\G(a)=2$, $\ord_\G(b)$ is an odd prime,
$\ord_\G(c)>1$ and coprime to both $2$ and $\ord_\G(b)$, and $c\circ b\circ a=e_\G$ for the \unit{} element $e_\G$ of $\G$.
Let $u,v\in\Sigma^\ast$ be such that $a=\fum_u$, $b=\fum_v$ and $c=(\fum_{uv})^-$,
and let $k=\ord_\G(\fum_{v})$ and $r=\ord_\G(c)=\ord_\G(\fum_{uv})$. Then $r>1$ and coprime to both $2$ and $k$.
Let $S=\bigl\{p\in \simclass{\Qr\!}\mid e_\G(p)=p\bigr\}$.
As $\fum_x$ is $\G$ for every $x\in\{u,v\}^\ast$, we have $e_\G\circ\fum_x=\fum_x$.
Thus,
\begin{align*}
& \fum_{xu^2}(q)=\fum_{u^2}\bigl(\fum_x(q)\bigr)=e_\G\bigl(\fum_x(q)\bigr)=(e_\G\circ\fum_x)(q)=\fum_x(q),\quad\mbox{and}\\
& \fum_{xv^k}(q)=\fum_{v^k}\bigl(\fum_x(q)\bigr)=e_\G\bigl(\fum_x(q)\bigr)=(e_\G\circ\fum_x)(q)=\fum_x(q),\quad\mbox{for every $q\in S$}.
\end{align*}
Then, by \eqref{groupini}, each of $\fum_u\!\!\restriction_S$, $\fum_v\!\!\restriction_S$
and $\fum_{uv}\!\!\restriction_S$ is a permutation on $S$.
By \eqref{groupinii}, the order of $\fum_u\!\!\restriction_S$ is $2$, the order of $\fum_v\!\!\restriction_S$ is $k$,
and  the order $l$ of $\fum_{uv}\!\!\restriction_S$ is a $>1$ divisor of $r$, and so it is coprime to both $2$ and $k$.
Also, we have $k,l\leq |S|\leq |Q|$.
Further, for every $x$, if $q$ is in $S$ then $\fum_x(q)\in S$ as well. So we have
\[
\fum_{x(uv)^l}(q)=\fum_{(uv)^l}\bigl(\fum_x(q)\bigr)=(\fum_{uv}\!\!\restriction_S)^l\bigl(\fum_x(q)\bigr)=\id_S\bigl(\fum_x(q)\bigr)=\fum_x(q),\ \ \mbox{for all $q\in S$}.
\]
It remains to show that there is $q\in S$ with $q\ne\fum_u(q)$, $q\ne\fum_u(q)$, and $q\ne\fum_{uv}(q)$.
 %
%
Recall that the length of any cycle in a permutation divides its order.
First, we show there is $q\in S$ with $q\ne\fum_u(q)$ and $q\ne\fum_u(q)$. Indeed, as
$\fum_{u}\!\!\restriction_S\ne\id_S$, there is $q\in S$ such that $\fum_u(q)=q'\ne q$.
As the order of $\fum_{u}\!\!\restriction_S$ is $2$, $\fum_u(q')=q$.
If both $\fum_v(q)=q$ and $\fum_v(q')=q'$ were the case, then $\fum_{uv}(q)=q'$ and $\fum_{uv}(q')=q$ would hold, and so
$(qq')$ would be a cycle in $\fum_{uv}\!\!\restriction_S$, contrary to $l$ being coprime to $2$.
So take some $q\in S$ with $\fum_u(q)=q'\ne q$ and  $\fum_v(q)\ne q$. 
If $\fum_v(q')\ne q$ then $\fum_{uv}(q)\ne q$, and so $q$ is a good choice.
Suppose $\fum_v(q')=q$, and let $q''=\fum_v(q)$. Then $q''\ne q'$, as $k$ is odd. Thus, $\fum_{uv}(q')\ne q'$, and
so $q'$ is a good choice.

$(iii)~(\Leftarrow)$
Suppose $u,v\in\Sigma^\ast$, $q\in \Qr$, and $k,l<\omega$ are satisfying the conditions.
For every $x\in\{u,v\}^\ast$, we define an equivalence relation $\approx_x$ on $\simclass{\Qr\!}$ by taking $p\approx_x p'$ iff
$\fum_x(p)=\fum_x(p')$. Then we clearly have that $\approx_x\subseteq \approx_{xy}$, for all $x,y\in \{u,v\}^\ast$.
As $Q$ is finite, there is $z\in \{u,v\}^\ast$ such that
$\approx_z= \approx_{zy}$ for all $y\in \{u,v\}^\ast$.
Take such a $z$. By \eqref{gini}, $\fum_z^n$ is idempotent for some $n\geq 1$. We let $w=z^n$.
Then $\fum_w$ is idempotent and we also have that
\begin{equation}\label{fp}
\approx_w\,=\, \approx_{wy}\quad\mbox{for all $y\in \{u,v\}^\ast$.}
\end{equation}
Let $G_{\{u,v\}}=\bigl\{\fum_{wxw}\mid x\in\{u,v\}^\ast\bigr\}$. Then $G_{\{u,v\}}$ is closed under composition. Let $\G_{\{u,v\}}$ be the subsemigroup of $M(\A_{\L(\A)})$  with universe $G_{\{u,v\}}$.
Then $\fum_w=\fum_{w\varepsilon w}$ is an \unit{} element in $\G_{\{u,v\}}$.
Let $S=\{p\in \simclass{\Qr\!}\mid \fum_w(p)=p\}$.
We show that
\begin{equation}\label{Sperm}
\mbox{for every $\fum$ in $\G_{\{u,v\}}$, $\fum\!\!\restriction_S$ is a permutation on $S$,}
\end{equation}
and so $\G_{\{u,v\}}$ is a group by \eqref{groupini}.
Indeed, take some $x\in\{u,v\}^\ast$. As
$\fum_w\bigl(\fum_{wxw}(p)\bigr)=\fum_{wxww}(p)=\fum_{wxw}(p)$ for any $p\in \simclass{\Qr\!}$, $\fum_{wxw}\!\!\restriction_S$ is an $S\to S$ function. Also, if $p,p'\in S$ and $\fum_{wxw}(p)=\fum_{wxw}(p')$ then $p\approx_{wxw}p'$. Thus, by \eqref{fp}, $p\approx_w p'$, that is,
$p=\fum_w(p)=\fum_w(p')=p'$, proving \eqref{Sperm}.

We show that $\G_{\{u,v\}}$ is unsolvable by finding an unsolvable homomorphic image of it.
Let $R=\bigl\{p\in \simclass{\Qr\!}\mid p=\fum_x(q)\mbox{ for some }x\in\{u,v\}^\ast\bigr\}$.
We claim that, for every $\fum$ in $\G_{\{u,v\}}$, $\fum\!\!\restriction_R$ is a permutation on $R$, and so the function
$h$ mapping every $\fum$ to $\fum\!\!\restriction_R$ is a group homomorphism from $\G_{\{u,v\}}$ to the group of all permutations on $R$. Indeed, by \eqref{Sperm}, it is enough to show that $R\subseteq S$.
Let $\overline{w}=\overline{z}_{m}\dots\overline{z}_1$, where $w=z_1\dots z_m$ for some $z_i\in\{u,v\}$,
$\overline{u}=u$ and $\overline{v}=v^{k-1}$.
Since $\fum_{x}(q)=\fum_{x(u)^{2}}(q)=\fum_{x(v)^{k}}(q)$ for all $x\in\{u,v\}^\ast$,
we obtain that
\begin{multline}\label{barw}
\fum_{yw\overline{w}}(q)=\fum_{\overline{z}_{m-1}\dots\overline{z}_1}\bigl(\fum_{yz_1\dots z_m\overline{z}_m}(q)\bigr)
=\fum_{\overline{z}_{m-1}\dots\overline{z}_1}\bigl(\fum_{yz_1\dots z_{m-1}}(q)\bigr)=\dots\\
\dots =\fum_{\overline{z}_1}\bigl(\fum_{yz_1}(q)\bigr)=\fum_{xz_1\overline{z}_1}(q)=\fum_y(q),\quad
\mbox{for all $y\in\{u,v\}^\ast$.}
\end{multline}
Now suppose $p\in R$, that is, $p=\fum_x(q)$ for some $x\in \{u,v\}^\ast$. Then, by \eqref{barw},
\[
\fum_w(p)=\fum_w\bigl(\fum_x(q)\bigr)=\fum_{xw}(q)=\fum_{xww\overline{w}}(q)=\fum_{xw\overline{w}}(q)=\fum_x(q)=p,
\]
and so $p\in S$, as required.

Now let $\G$ be the image of $\G_{\{u,v\}}$ under $h$.
We prove that $\G$ is unsolvable by finding three elements $a,b,c$ in it such that $\ord_\G(a)=2$, $\ord_\G(b)=k$,
$\ord_\G(c)$ is coprime to both $2$ and $\ord_\G(b)$, and $c\circ b\circ a=\id_R$  (the \unit{} element of $\G$).
So let $a=h(\fum_{wuw})$,  $b=h(\fum_{wvw})$, and $c=h(\fum_{wuvw})^-$.
Observe that, for every $x\in\{u,v\}^\ast$, $h(\fum_{wxw})=\fum_x\!\!\restriction_R$, and so $c\circ b\circ a=\id_R$.
Also, for any $\fum_x(q)\in R$, $a^2\bigl(\fum_x(q)\bigr)=(\fum_u\!\!\restriction_R)^2\bigl(\fum_x(q)\bigr)=\fum_{xu^2}(q)=\fum_x(q)$
by our assumption, so $a^2=\id_R$. On the other hand, $q\in R$ as $\fum_\varepsilon(q)=q$, and
$\id_R(q)=q\ne\fum_u(q)$ by assumption, so $a\ne\id_R$.  As $\ord_\G(a)$ divides $2$, $\ord_\G(a)=2$ follows.
Similarly, we can show that $\ord_\G(b)=k$ (using that $\fum_{xv^k}(q)=\fum_x(q)$ for every $x\in\{u,v\}^\ast$, and $u\ne\fum_v(q)$). Finally (using that $\fum_{x(uv)^l}(q)=\fum_x(q)$ for every $x\in\{u,v\}^\ast$, and $u\ne\fum_{uv}(q)$),
we obtain that $h(\fum_{wuvw})^l=\id_R$ and $h(\fum_{wuvw})\ne\id_R$.
Therefore, it follows that $\ord_\G(c)=\ord_\G\bigl(h(\fum_{wuvw})^-\bigr)=\ord_\G\bigl(h(\fum_{wuvw})\bigr)>1$ and divides $l$, and so coprime to both $2$ and $k$, as required.
\end{proof}

\section{Deciding FO-definability: \PSpace-hardness}\label{sec:reglang}

Kozen~\cite{Kozen77} showed that deciding whether the intersection of the languages recognised by a set of given deterministic DFAs is non-empty is \PSpace-complete. By carefully analysing Kozen's lower bound proof and using the 
criterion of Theorem~\ref{DFAcrit}~$(i)$, Cho and Huynh~\cite{DBLP:journals/TCS/ChoHyunh91} established that
deciding $\FO(<)$-definability of $\L(\A)$ is \PSpace-hard, for any given minimal DFA $\A$. We generalise their construction and use the criteria in Theorem~\ref{DFAcrit}~$(ii)$--$(iii)$ to cover $\FO(<,\equiv)$- and $\FO(<,\MOD)$-definability as well. 

\begin{theorem}\label{DFAhard}
For any $\lang \in \{ \FO(<), \FO(<,\equiv), \FO(<,\MOD)\}$,  deciding $\lang$-defina\-bi\-lity of the language $\L(\A)$ of a given minimal DFA $\A$ is \PSpace-hard.
\end{theorem}
\begin{proof}
Let $\M$ be a deterministic Turing machine that decides a language using at most $\ppn=P_{\M}(n)$ tape cells on any input of size $n$, for some polynomial $P_{\M}$. Given such an $\M$ and an input $\boldsymbol{x}$,
our aim is to define three minimal DFAs whose languages are, respectively, $\FO(<)$-, $\FO(<,\equiv)$-, and $\FO(<,\MOD)$-definable iff $\M$ rejects $\boldsymbol{x}$, and whose sizes are polynomial in $\ppn$ and the size $|\M|$ of $\M$.


Suppose $\M =(Q, \Tmabc, \Tmtran, \B, q_0, \qa)$ with a set $Q$ of states, tape alphabet $\Tmabc$ with $\B$ for blank, transition function $\Tmtran$, initial state $q_0$ and accepting state $\qa$.
Without loss of generality we assume that
$\M$ erases the tape before accepting, its head is at the left-most cell in an accepting configuration, 
and if $\M$ does not accept the input, it runs forever.
Given an input word $\boldsymbol{x}=x_1\dots x_n$ over $\Tmabc$,
we represent configurations $\conf$ of the computation of $\M$ on $\boldsymbol{x}$ by
the $\ppn$-long word written on the tape (with sufficiently many blanks at the end) in which the symbol $y$ in the active cell is replaced by the pair $(q,y)$ for the current state $q$.
The accepting computation of $\M$ on $\boldsymbol{x}$ is encoded by a word
$\sharp\, \conf_1 \, \sharp \, \conf_2 \, \sharp\,  \dots \, \sharp \, \conf_{k-1} \, \sharp \, \conf_{k} \flat$ over
the alphabet $\Aiabc=\Tmabc\cup(Q\times\Tmabc)\cup\{\sharp,\flat\}$,
with $\conf_1,\conf_2,\dots,\conf_k$ being the subsequent configurations. In particular, $\conf_1$ is the initial configuration on $\boldsymbol{x}$ (so it is of the form $(q_0,x_1)x_2\dots x_n\B\dots\B$), and $\conf_k$ is the accepting configuration (so it is of the form $(\qa,\B)\B\dots\B$).
As usual for this representation of computations, we may regard $\Tmtran$ as a partial function
from  $\bigl(\Tmabc\cup(Q\times\Tmabc)\cup\{\sharp\}\bigr)^3$ to $\Tmabc\cup(Q\times\Tmabc)$ with $\Tmtran(\sigma^j_{i-1},\sigma^j_i,\sigma^j_{i+1})=\sigma^{j+1}_i$ for each $j< k$, where $\sigma^j_i$ is the $i$th symbol of $\conf^j$.

%
%

Let $p_{\M,\boldsymbol{x}}=\pp$ be the first prime such that $\pp\geq\ppn+2$ and $\pp\not\equiv\pm 1\ (\text{mod}\ 10)$. By~\cite[Corollary 1.6]{Illin2018}, $\pp$ is polynomial in $\ppn$.  
Our first aim is to construct a $p+1$-long sequence $\A_i$ of disjoint minimal DFAs over $\Aiabc$.
Each $\A_i$ has size polynomial in $\ppn$ and $|\M|$,
and it
checks certain properties of an accepting computation on $\boldsymbol{x}$ such that
$\M$ accepts $\boldsymbol{x}$ iff the intersection of the $\L(\A_i)$ is not empty and
consists of the single word encoding the accepting computation on $\boldsymbol{x}$.


We define each $\A_i$ as an NFA, and assume that it can be turned to a DFA by adding a `trash state' $\tst_i$
looping on itself with every $\sigma\in\Aiabc$, and adding the missing transitions leading to $\tst_i$.
The DFA $\A_{0}$ checks that an input starts with the initial configuration on $\boldsymbol{x}$ and ends with the accepting configuration:\\[2pt]
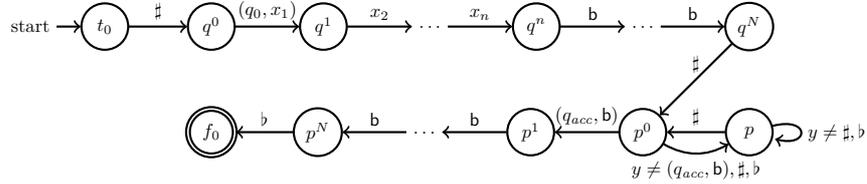
\centerline{
\begin{tikzpicture}[->,thick,scale=0.7,node distance=2cm, transform shape]
\node[state, initial] (1) {$t_0$};
\node[state, right of  =1] (2) {$q^0$};
\node[state, right  =1.2cm of 2] (3) {$q^1$};
\node[right  of= 3](4){$\ldots$};
\node[state, right  of= 4] (5) {$q^{n}$};
\node[right of=5](6){$\ldots$};
\node[state, right  of= 6] (7) {$q^{\ppn}$};
\node[state, below of= 7] (8) {$p$};
\node[state,left of= 8](9){$p^0$};
\node[state, left =1.2cm of 9] (10) {$p^1$};
\node[left of=10](11){$\ldots$};
\node[state, left of= 11] (12) {$p^{\ppn}$};
\node[state, accepting, left of =12] (f) {$f_0$};
\draw (1) edge[above] node{$\sharp$} (2)
(2) edge[above] node{$(q_0,x_1)$} (3)
(3) edge[above] node{$x_2$} (4)
(4) edge[above] node{$x_n$} (5)
(5) edge[above] node{$\B$} (6)
(6) edge[above] node{$\B$} (7)
(8) edge [ loop right ] node{$y\neq\sharp,\flat$} (8)
(7) edge[above] node{$\sharp$} (9)
(9) edge[above] node{$(\qa,\B)$} (10)
(8) edge[above] node{$\sharp$} (9)
(9) edge[bend right,below] node{$y\neq (\qa,\B),\sharp,\flat$} (8)
(10) edge[above] node{$\B$} (11)
(11) edge[above] node{$\B$} (12)
(12) edge[above] node{$\flat$} (f)
;
\end{tikzpicture}
}\\
When $1 \le i \le \ppn$,
the DFA $\A_{i}$ checks, for all $j< k$, whether the $i$th symbol of $\conf^j$ changes `according to $\Tmtran$' in passing to $\conf^{j+1}$. The non-trash part of its transition function $\delta^i$ is as follows, for $1<i<N$. (For $i=1$ and $i=N$ some adjustments are needed.) For all $u,u',v,w,w',y,z\in\Tmabc\cup(Q\times\Tmabc)$, 
\begin{align*}
& \delta^i_\sharp(t_i)=q^0,\quad
\delta^i_u(q^j)=q^{j+1},\ \mbox{for $j=0,...,i-3$,}\quad
\delta^i_u(q^{i-2})=r_u,\quad
\delta^i_v(r_u)=r_{uv},\\
& \delta^i_w(r_{uv})=q^0_{\Tmtran(u,v,w)},\quad
 \delta^i_y(q^j_z)=q^{j+1}_z,\ \mbox{for $j=0,...,N-3$, $j\ne N-i-1$,}\\
&  \delta^i_\sharp(q^{N-i-1}_z)=q^{N-i}_z,\quad
\delta^i_\flat(q^{N-i-1}_z)=f_i,\quad
\delta^i_{u'}(q^{N-2}_z)=p_{u'z},\ \
\delta^i_z(p_{u'z})=r_{u'z},\\
& \mbox{see below, where $z=\Tmtran(u,v,w)$ and $z'=\Tmtran(u',z,w')$:}
\end{align*}
\\[-.4cm]
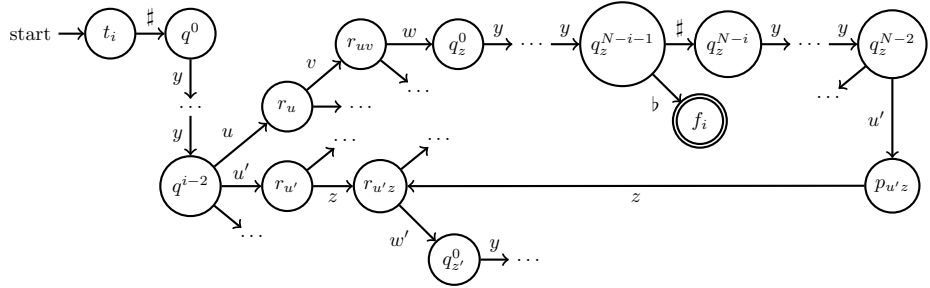
\centerline{
\begin{tikzpicture}[->,thick,scale=0.75,node distance=2cm, transform shape]
\node[state, initial] 	(ti) 	{$t_i$};
\node[state, right=5mm of ti]  (q0)  {$q^0$};
\node[below =7mm of q0]			(dotsq0)	{$\ldots$};
\node[state,below =7mm of dotsq0] 	(qi2) 	{$q^{i-2}$};
\node[below right=5mm of  qi2]	(dotsqi2) {$\ldots$};
\node[state, right=7mm of qi2]  (rup)  {$r_{u'}$};
\node[above right=5mm of rup]  (dotsrup)  {$\ldots$};
\node[state, above=5mm of  rup]  (ru)  {$r_{u}$};
\node[state, right=7mm of rup]  (rupz)  {$r_{u'z}$};
\node[above right=5mm of rupz]  (dotsrupz)  {$\ldots$};
\node[state, below right=9mm of rupz]  (qzp0)  {$q_{z'}^{0}$};
\node[right=5mm of qzp0]  (dotsqzp0)  {$\ldots$};
\node[right=5mm of ru]  (dotsru)  {$\ldots$};
\node[state, above=5mm of  dotsru]  (ruv)  {$r_{uv}$};
\node[below right=5mm of ruv]  (dotsruv)  {$\ldots$};
\node[state, right=8mm of ruv]  (qz0)  {$q_{z}^{0}$};
\node[right=5mm of qz0]  (dotsqz0)  {$\ldots$};
\node[state, right=5mm of dotsqz0]  (qzNi1)  {$q_{z}^{N-i-1}$};
\node[state, right=5mm of qzNi1]  (qzNi)  {$q_{z}^{N-i}$};
\node[right=5mm of qzNi]  (dotsqzNi)  {$\ldots$};
\node[state, right=5mm of dotsqzNi]  (qzN2)  {$q_{z}^{N-2}$};
\node[below left=5mm of  qzN2]	(dotsqzN2) {$\ldots$};
\node[state, below=14mm of qzN2]  (pupz)  {$p_{u'z}$};
\node[state,accepting, below right=7mm of qzNi1] (fi) {$f_i$};
\draw
(ti) 	edge[above] 	node{$\sharp$}  (q0)
(q0) 	edge[left]	 node{$y$} 	(dotsq0)
(dotsq0) 	edge[left]	 node{$y$} 	(qi2)
(qi2) 	edge[above] 	node{$u'$}  (rup)
(qi2) edge[above] 	node{$$}  (dotsqi2)
(rup) edge[below] 	node{$z$}  (rupz)
(rup) edge[above] 	node{$$}  (dotsrup)
(rupz) edge[above] 	node{$$}  (dotsrupz)
(rupz) edge[below left] 	node{$w'$}  (qzp0)
(qzp0) edge[above] 	node{$y$}  (dotsqzp0)
(qi2) 	edge[above left] 	node{$u$}  (ru)
(ru) edge[above] 	node{$$}  (dotsru)
(ru) 	edge[above left] 	node{$v$}  (ruv)
(ruv) 	edge[above] 	node{$w$}  (qz0)
(ruv) edge[above] 	node{$$}  (dotsruv)
(qz0) edge[above] 	node{$y$}  (dotsqz0)
(dotsqz0) 	edge[above] node{$y$}  (qzNi1)
(qzNi1) edge[above] 	node{$\sharp$}  (qzNi)
(qzNi) edge[above] 	node{$y$}  (dotsqzNi)
(dotsqzNi) edge[above] node{$y$}  (qzN2)
(qzN2) edge[left] node{$u'$}  (pupz)
(qzN2) edge[above] 	node{$$}  (dotsqzN2)
(pupz) edge[below] 	node{$z$}  (rupz)
(qzNi1) edge[below left] node{$\flat$}  (fi)
;
\end{tikzpicture}
}
\\
%
%
Finally, if $\ppn+1\le i\le \pp$  then $\A_i$ accepts all words over $\Aiabc$ with a single occurrence of $\flat$, which is the input's last character:\\
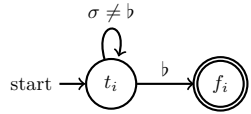
\centerline{
\begin{tikzpicture}[->,thick,scale=0.75, transform shape]
\node[state, initial] (1) {$t_i$};
\node[state, accepting, right = 1cm of  1] (2) {$f_i$};
\draw (1) edge[loop above] node{$\sigma\neq\flat$} (1)
(1) edge[above] node{$\flat$} (2)
;
\end{tikzpicture}
}
Note that $\A_{\pp-1}=\A_\pp$ as $\pp\geq\ppn+2$.
It is not hard to check that each of the $\A_i$ is a minimal DFA that does not contain nontrivial cycles and
the following holds:

\begin{lemma}\label{l:DFAs}
$\M$ accepts $\boldsymbol{x}$ iff $\bigcap_{i=0}^{\pp} \L(\A_i) \ne \emptyset$, in which case this language consists of a single word that encodes the accepting computation of $\M$ on $\boldsymbol{x}$.
%
\end{lemma}


Next, we require three sequences of DFAs $\aut^p_<$, $\aut^p_\equiv$ and $\aut^p_\MOD$, where $p>5$ is a prime number with  $p\not\equiv\pm 1\ (\text{mod}\ 10)$; see the picture below for $p=7$.\\[2pt]
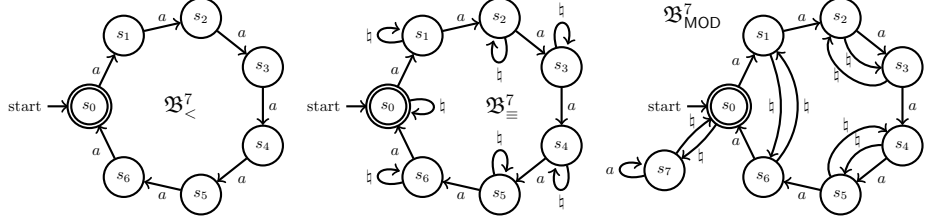
\centerline{
\begin{tikzpicture}[->,thick,scale=0.6,node distance=2cm, transform shape]
  \node[state,initial,accepting] (0)at (180:2cm) {$s_0$};
  \foreach \x in {1,...,6}{%
    \pgfmathparse{(-\x)*(360/7)+180}]
    \node[state] (\x) at (\pgfmathresult:2cm) {$s_\x$};
  }
 \draw (0) edge[left] node{$a$} (1)
(1) edge[above] node{$a$} (2)
(2) edge[above right] node{$a$} (3)
(3) edge[right] node{$a$} (4)
(4) edge[below right] node{$a$} (5)
(5) edge[below] node{$a$} (6)
(6) edge[below left] node{$a$} (0)
;
\node (0) at (0,0) {\Large $\aut^7_<$};
\end{tikzpicture}
\
\begin{tikzpicture}[->,thick,scale=0.6,node distance=2cm, transform shape]
  \node[state,initial,accepting] (0)at (180:2cm) {$s_0$};
  \foreach \x in {1,...,6}{%
    \pgfmathparse{(-\x)*(360/7)+180}
    \node[state] (\x) at (\pgfmathresult:2cm) {$s_\x$};
  }
 \draw (0) edge[left] node{$a$} (1)
 (0) edge[loop right,right] node{$\natural$}(0)
(1) edge[above] node{$a$} (2)
(1) edge[loop left,left] node{$\natural$}(1)
(2) edge[above right] node{$a$} (3)
(2) edge[loop below,below] node{$\natural$}(2)
(3) edge[right] node{$a$} (4)
(3) edge[loop above,above] node{$\natural$}(3)
(4) edge[below right] node{$a$} (5)
(4) edge[loop below,below] node{$\natural$}(4)
(5) edge[below] node{$a$} (6)
(5) edge[loop above,above] node{$\natural$}(5)
(6) edge[below left] node{$a$} (0)
(6) edge[loop left,left] node{$\natural$}(6)
;
\node (0) at (0.5,0) {\Large$\aut^7_\equiv$};
\end{tikzpicture}
\
\begin{tikzpicture}[->,thick,scale=0.6,node distance=2cm, transform shape]
 \node[state,initial,accepting] (0)at (180:2cm) {$s_0$};
  \foreach \x in {1,...,6}{%
    \pgfmathparse{(-\x)*(360/7)+180}
    \node[state] (\x) at (\pgfmathresult:2cm) {$s_\x$};
  }
  \node[state, below left of =0] (7){$s_7$};
\draw (0) edge[above left] node{$a$} (1)
(0) edge[bend left=10,below] node{$\natural$} (7)
(1) edge[above] node{$a$} (2)
(1) edge[bend left=20,left] node{$\natural$} (6)
(2) edge[above right] node{$a$} (3)
(2) edge[bend right=40,left] node{$\natural$} (3)
(3) edge[right] node{$a$} (4)
(3) edge[bend left=80,left] node{$\natural$} (2)
(4) edge[below right] node{$a$} (5)
(4) edge[bend right=40,above] node{$\natural$} (5)
(5) edge[below] node{$a$} (6)
(5) edge[bend left=80,above] node{$\natural$} (4)
(6) edge[left] node{$a$} (0)
(6) edge[bend right=40,right] node{$\natural$} (1)
(7) edge[loop left,left] node{$a$} (7)
(7) edge[bend left=10,above ] node{$\natural$} (0)
;
\node (0) at (-2.8,2) {\Large$\aut^7_\MOD$};
\end{tikzpicture}
}\\[2pt]
In general, the first sequence
is $\aut^p_< = \bigl(\{s_{i} \mid i < p \},\{a\},\perm{\aut}{<}{p},s_0,\{s_0\} \bigr)$, where
$\perm{\aut}{<}{p}_a(s_i)=s_j$ if $i,j<p$ and $j\equiv i+1\ (\text{mod}\ p)$.
%
%
Then $\L(\aut^p_<)$ comprises all words of the form $(a^{p})^\ast$,
$\aut^p_<$ is the minimal DFA for $\L(\aut^p_<)$, and the syntactic monoid $M(\aut^p_<)$
 is the cyclic group of order $p$ (generated by the permutation $\smash{\perm{\aut}{<}{p}_a}$).

The second sequence is
%
$\aut^p_\equiv = \bigl(\{s_{i} \mid i < p \},\{a,\natural\},\perm{\aut}{\equiv}{p},s_0,\{s_0\} \bigr)$, where
$\perm{\aut}{\equiv}{p}_\natural(s_i)=s_i$ and
$\perm{\aut}{\equiv}{p}_a(s_i)=s_j$ if $i,j<p$ and $j\equiv i+1\ (\text{mod}\ p)$.
One can check that $\L(\aut^p_\equiv)$ comprises all words of $a$'s and $\natural$'s
where the number of $a$'s is divisible by $p$,
$\aut^p_\equiv$ is the minimal DFA for this language, and $M(\aut^p_\equiv)$
 is also the cyclic group of order $p$ (generated by the permutation $\perm{\aut}{\equiv}{p}_a$).

%
%
%
%
%
%
%

The third sequence is
%
%
$\aut^p_\MOD = \bigl(\{s_{i} \mid i \leq p \},\{a,\natural\},\perm{\aut}{\MOD}{p},s_0,\{s_0\} \bigr)$, where
\begin{itemize}
\item
$\perm{\aut}{\MOD}{p}_a(s_p)=s_p$, and
$\perm{\aut}{\MOD}{p}_a(s_i)=s_j$ whenever $i,j<p$ and $j\equiv i+1\ (\text{mod}\ p)$;

\item
$\perm{\aut}{\MOD}{p}_\natural(s_0)=s_p$, $\perm{\aut}{\MOD}{p}_\natural(s_p)=s_0$, and
$\perm{\aut}{\MOD}{p}_\natural(s_i)=s_j$ whenever $1\leq i,j<p$ and $i\cdot j\equiv p-1\ (\text{mod}\ p)$,
that is, $j = -1/i$ in the finite field $\mathbb F_p$.
\end{itemize}
One can check that $\aut^p_\MOD$ is the minimal DFA for its language,
and the syntactic monoid $M(\aut^p_\MOD)$ is the permutation group generated by
$\perm{\aut}{\MOD}{p}_a$ and $\perm{\aut}{\MOD}{p}_\natural$.

\begin{lemma}\label{algebralemma}
For any prime $p>5$ with $p\not\equiv\pm 1\ (\text{mod}\ 10)$, the group $M(\aut^p_\MOD)$ is unsolvable, but all of its proper subgroups are solvable.

%
\end{lemma}
\begin{proof}
One can check that the order of the permutation $\perm{\aut}{\MOD}{p}_\natural$ is $2$, that of $\perm{\aut}{\MOD}{p}_a$ is $p$, while the order of the inverse of $\perm{\aut}{\MOD}{p}_{\natural a}$ is the same as the
order of $\perm{\aut}{\MOD}{p}_{\natural a}$, which is $3$. So $M(\aut^p_\MOD)$ is unsolvable, for any prime $p$,
by the Kaplan--Levy criterion.
To prove that all proper subgroups of $M(\aut^p_\MOD)$ are solvable,
we show that $M(\aut^p_\MOD)$ is a subgroup of the
\emph{projective special linear group} $\text{\sc PSL}_2(p)$. If $p$ is a prime with $p>5$ and $p\not\equiv\pm 1\ (\text{mod}\ 10)$, then
all proper subgroups of $\text{\sc PSL}_2(p)$ are solvable; see, e.g.,~\cite[Theorem~2.1]{DBLP:conf/bcc/King05}.
(So $M(\aut^p_\MOD)$  is in fact isomorphic to the unsolvable group $\text{\sc PSL}_2(p)$.)
Consider the set $P=\{0,1,\dots,p-1,\infty\}$ of all points of the projective line over the field $\mathbb F_p$.
By identifying $s_i$ with $i$ for $i<p$, and $s_p$ with $\infty$,
we may regard the elements of $M(\aut^p_\MOD)$ as $P\to P$ functions.
The group $\text{\sc PSL}_2(p)$ consists of all $P\to P$ functions of the form
$i\mapsto\frac{w\cdot i + x}{y\cdot i+z}$, 
where $w\cdot z-x\cdot y=1$,  with the field arithmetic of $\mathbb F_p$ extended by $i+\infty=\infty$ for any $i\in P$, $0\cdot\infty=1$ and $i\cdot\infty=\infty$ for $i\ne 0$.
%
%
One can check that the two generators of $M(\aut^p_\MOD)$ are in $\text{\sc PSL}_2(p)$: take
$w=1$, $x=1$, $y=0$, $z=1$ for $\perm{\aut}{\MOD}{p}_a$,
and $w=0$, $x=1$, $y=p-1$, $z=0$ for $\perm{\aut}{\MOD}{p}_\natural$.
\end{proof}




Finally, we 
define three automata $\A_<$, $\A_\equiv$, $\A_\MOD$ over the same tape alphabet
$\Aabc = \Aiabc\cup\{a_1,a_2,\natural\}$, where $a_1,a_2$ are fresh symbols.
We take, respectively,  $\aut^{\pp}_<$, $\aut^{\pp}_\equiv$, $\aut^{\pp}_\MOD$  and replace each transition $s_i\to_a s_j$ in them by a fresh copy of $\A_i$, for $i \le \pp$, as shown in the picture below.\\
%
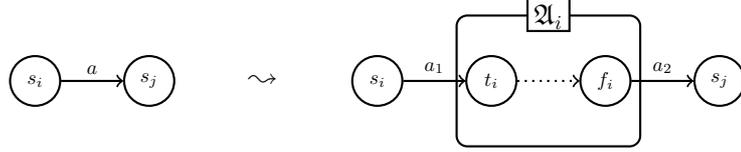
\centerline{
\begin{tikzpicture}[->,thick,scale=0.75,node distance=2cm, transform shape]
\node[state] (si) {$s_i$};
\node[state, right  of =si] (sj) {$s_{j}$};
\node[right  of= sj] (maps) {\Large\bf$\leadsto$};
\node[state, right  of= maps] (si2) {$s_i$};
\node[state, right of=si2] (q) {$t_i$};
\node[state,right of= q](f){$f_i$};
\node[state, right of =f] (sj2) {$s_j$};
\node[fit = (q)(f), basic box = black, header = $\A_i$] (A) {};
\draw (si) edge [above] node{$a$} (sj)
(si2) edge [above] node{$a_1$} (q)
(f) edge [above] node{$a_2$} (sj2)
(q) edge [dotted,above] node{$ $} (f)
;
\end{tikzpicture}
}\\
We make $\A_<$, $\A_\equiv$, $\A_\MOD$ deterministic by adding a trash state $\tst$
looping on itself with every $y\in\Aabc$, and adding the missing transitions leading to $\tst$.
It follows that $\A_<$, $\A_\equiv$, and $\A_\MOD$ are minimal DFAs of size polynomial in $\ppn$, $|\M|$.

\begin{lemma}\label{l:empty}
$(i)$ $\L(\A_<)$ is $\FO(<)$-definable iff $\bigcap_{i=0}^{\pp} \L(\A_i)= \emptyset$.

$(ii)$ $\L(\A_\equiv)$ is $\FO(<,\equiv)$-definable iff $\bigcap_{i=0}^{\pp} \L(\A_i)= \emptyset$.

$(iii)$ $\L(\A_\MOD)$ is $\FO(<,\MOD)$-definable iff  $\bigcap_{i=0}^{\pp} \L(\A_i)= \emptyset$.
\end{lemma}
\begin{proof}
As $\A_<,\A_\equiv,\A_\MOD$ are minimal, we can replace $\simm$ by $=$ in the conditions
of Theorem~\ref{DFAcrit}.
For the ($\Rightarrow$) directions, given some  $w\in\bigcap_{i=0}^{\pp} \L(\A_i) $, in each case we show how to satisfy the corresponding condition of Theorem~\ref{DFAcrit}:
$(i)$ take $u=a_1wa_2$, $q=s_0$, and $k=\pp$;
$(ii)$ take $u=a_1wa_2$, $v=\natural^{|u|}$, $q=s_0$, and $k=\pp$;
$(iii)$ take $u=\natural$,  $v=a_1wa_2$, $q = s_0$, $k = \pp$ and $l = 3$.

$(\Leftarrow)$ We show that the corresponding condition of Theorem~\ref{DFAcrit}
implies non-emptiness of $\bigcap_{i=0}^{\pp} \L(\A_i) $.
To this end, we define a
$\Aabc^\ast\to\{a,\natural\}^\ast$ homomorphism by taking $h(\natural)=\natural$, $h(a_1)=a$, and $h(b)=\varepsilon$ for all other $b\in\Aabc$.

\smallskip
$(i)$ and $(ii)$: Let $\circ\in\{<,\equiv\}$ and suppose $q$ is a state in $\A^\pp_\circ$ and $u'\in\Aabc^\ast$ such that $q\ne\perm{\A}{\circ}{\pp}_{u'}(q)$ and
$q=\perm{\A}{\circ}{\pp}_{(u')^k}(q)$ for some $k$.
Let $S=\{s_0,s_1,\dots,s_{\pp-1}\}$.
 We claim that there exist $s\in S$ and $u\in \Aabc^\ast$ such that
 \begin{align}
 \label{unotid}
 & s\ne\perm{\A}{\circ}{\pp}_{u}(s),\\
 \label{uins}
 & \perm{\A}{\circ}{\pp}_{x}(s)\in S,\quad\mbox{for every $x\in\{u\}^\ast$.}
 \end{align}
 Indeed, observe that none of the states along the cyclic $q\to_{(u')^k} q$ path $\Pi$ in $\A^\pp_\circ$ is $\tst$.
 So there is some state along $\Pi$ that is in $S$, as otherwise one of the $\A_i$ would contain a nontrivial cycle.
 Therefore, $u'$ must be of the form $w\natural^n a_1w'$ for some $w\in\Aiabc^\ast$, $n<\omega$ and $w'\in\Aabc^\ast$.
 It is easy to see that $s=\perm{\A}{\circ}{\pp}_{(u')^{k-1}w}(q)$ and $u=\natural^na_1w'w$ is as required in \eqref{unotid} and \eqref{uins}.

 As $M(\aut^\pp_\circ)$ is a finite group, the set $\bigl\{\perm{\aut}{\circ}{\pp}_{h(x)}\mid x\in\{u\}^\ast\bigr\}$
 forms a subgroup $\G$ in it (the subgroup generated by $\perm{\aut}{\circ}{\pp}_{h(u)}$).
 We show that $\G$ is nontrivial by finding a nontrivial homomorphic image of it.
 To this end,
  \eqref{uins} implies that, for every $x\in\{u\}^\ast$, the restriction $\perm{\A}{\circ}{\pp}_{x}\!\!\restriction_{S'}$ of
 $\perm{\A}{\circ}{\pp}_{x}$ to the set
 $S'=\bigl\{ \perm{\A}{\circ}{\pp}_{y}(s)\mid y\in\{u\}^\ast\bigr\}$ is an $S'\to S'$ function and
 $\perm{\A}{\circ}{\pp}_{x}\!\!\restriction_{S'}=\perm{\aut}{\circ}{\pp}_{h(x)}\!\!\restriction_{S'}$.
As  $M(\aut^\pp_\circ)$ is a group of permutations on a set containing $S'$,
$\perm{\aut}{\circ}{\pp}_{h(x)}\!\!\restriction_{S'}$ is a permutation of $S'$, for every $x\in\{u\}^\ast$.
Thus, $\bigl\{\perm{\aut}{\circ}{\pp}_{h(x)}\!\!\restriction_{S'}\mid x\in\{u\}^\ast\bigr\}$ is a homomorphic image of $\G$
that is nontrivial by \eqref{unotid}.

As $\G$ is a nontrivial subgroup of the cyclic group $M(\aut^\pp_\circ)$ of order $\pp$ and $\pp$ is a prime, $\G=M(\aut^\pp_\circ)$. Then there is $x\in\{u\}^\ast$ with $\perm{\aut}{\circ}{\pp}_{h(x)}=\perm{\aut}{\circ}{\pp}_{a}$ (a permutation containing the $\pp$-cycle $(s_0 s_1\dots s_{\pp-1})$ `around' all elements of $S$), and so $S'=S$
and $x=\natural^n a_1wa_2w'$ for some $n<\omega$, $w\in\Aiabc^\ast$, and $w'\in\Aabc^\ast$.
As $n=0$ when $\circ=<$ and $\perm{\A}{\equiv}{\pp}_{\natural^n}(s)$ for every $s\in S$,
$S'=S$ implies that $w\in\bigcap_{i=0}^{\pp-1} \L(\A_i)=\bigcap_{i=0}^{\pp} \L(\A_i)$.

$(iii)$ Suppose $q$ is a state in $\A^\pp_\MOD$ and $u',v'\in\Aabc^\ast$ such that
$q\ne\perm{\A}{\MOD}{\pp}_{u'}(q)$, $q\ne\perm{\A}{\MOD}{\pp}_{v'}(q)$, $q\ne\perm{\A}{\MOD}{\pp}_{u'v'}(q)$,
and $\perm{\A}{\MOD}{\pp}_{x}(q)=\perm{\A}{\MOD}{\pp}_{x(u')^2}(q)=\perm{\A}{\MOD}{\pp}_{x(v')^k}(q)=\perm{\A}{\MOD}{\pp}_{x(u'v')^l}(q)$ for some odd prime $k$ and number $l$ that is coprime to both $2$ and $k$.
Take $S=\{s_0,s_1,\dots,s_{\pp}\}$.
 We claim that there exist $s\in S$ and $u,v\in \Aabc^\ast$ such that
 \begin{align}
 \label{allnotid}
 & s\ne\perm{\A}{\MOD}{\pp}_{u}(s),\ s\ne\perm{\A}{\MOD}{\pp}_{v}(s),\ s\ne\perm{\A}{\MOD}{\pp}_{uv}(s),\\
  \label{allins}
 & \perm{\A}{\MOD}{\pp}_{x}(s)\in S,\quad\mbox{for every $x\in\{u,v\}^\ast$,}\\
  \label{allorder}
 & \perm{\A}{\MOD}{\pp}_{x}(s)=\perm{\A}{\MOD}{\pp}_{xu^2}(s)=\perm{\A}{\MOD}{\pp}_{xv^k}(s)=\perm{\A}{\MOD}{\pp}_{x(uv)^l}(s),\quad\mbox{for every $x\in\{u,v\}^\ast$.}
 \end{align}
  Indeed, by an argument similar to the one in the proof of $(i)$ and $(ii)$ above, we must have $u'=w_u\natural^n a_1w'_u$ and $v'=w_v\natural^m a_1w'_v$ for some $w_u,w_v\in\Aiabc^\ast$, $n,m<\omega$ and $w'_u,w'_v\in\Aabc^\ast$.
 For every $x\in\{u,v\}^\ast$, as both
  $\perm{\A}{\MOD}{\pp}_{xw_u}(q)$ and $\perm{\A}{\MOD}{\pp}_{xw_v}(q)$
 are in $S$, they must be the same state. Using this
it is not hard to see that $s=\perm{\A}{\MOD}{\pp}_{u'w_u}(q)$, $u=\natural^na_1w'_uw_u$ and $v=\natural^ma_1w'_vw_v$ are as required in \eqref{allnotid}--\eqref{allorder}.

 As $M(\aut^\pp_\MOD)$ is a finite group, the set $\bigl\{\perm{\aut}{\MOD}{\pp}_{h(x)}\mid x\in\{u,v\}^\ast\bigr\}$
 forms a subgroup $\G$ in it (the subgroup generated by $\perm{\aut}{\MOD}{\pp}_{h(u)}$ and $\perm{\aut}{\MOD}{\pp}_{h(v)}$). We show that $\G$ is unsolvable by finding an unsolvable homomorphic
 image of it. To this end, we let $S'=\bigl\{ \perm{\A}{\MOD}{\pp}_{y}(s)\mid y\in\{u,v\}^\ast\bigr\}$.
 Then \eqref{allins} implies that $S'\subseteq S$ and
 \begin{equation}\label{sclosed}
 \perm{\aut}{\MOD}{\pp}_{h(x)}(s')=\perm{\A}{\MOD}{\pp}_{x}(s')\in S',\quad
 \mbox{for all $s'\in S$ and $x\in\{u,v\}^\ast$,}
 \end{equation}
 and so the restriction $\perm{\A}{\MOD}{\pp}_{x}\!\!\restriction_{S'}$ of
 $\perm{\A}{\MOD}{\pp}_{x}$ to $S'$ is an $S'\to S'$ function and
 $\perm{\A}{\MOD}{\pp}_{x}\!\!\restriction_{S'}=\perm{\aut}{\MOD}{\pp}_{h(x)}\!\!\restriction_{S'}$.
 As  $M(\aut^\pp_\MOD)$ is a group of permutations on a set containing $S'$,
$\perm{\aut}{\MOD}{\pp}_{h(x)}\!\!\restriction_{S'}$ is a permutation of $S'$, for any  $x\in\{u,v\}^\ast$.
So \mbox{$\{\perm{\aut}{\MOD}{\pp}_{h(x)}\!\!\restriction_{S'}\mid x\in\{u,v\}^\ast\!\}$} is a homomorphic image of $\G$ that
is unsolvable by the Kaplan--Levy criterion: By \eqref{allnotid}, \eqref{allorder}, and $2$ and $k$ being primes,
the order of the permutation $\perm{\aut}{\MOD}{\pp}_{h(u)}\!\!\restriction_{S'}$ is $2$,
the order of $\perm{\aut}{\MOD}{\pp}_{h(v)}\!\!\restriction_{S'}$ is $k$,
and the order of $\perm{\aut}{\MOD}{\pp}_{h(uv)}\!\!\restriction_{S'}$ (which is the same as the order of its inverse)
is a $>1$ divisor of $l$, and so coprime to both $2$ and $k$.

As $\G$ is an unsolvable subgroup of $M(\aut^\pp_\MOD)$, it follows from Lemma~\ref{algebralemma} that
$\G=M(\aut^\pp_\MOD)$, and so $\{u,v\}^\ast\not\subseteq\natural^\ast$. We claim that $S'=S$ also follows.
Indeed, let $x\in\{u,v\}^\ast$ be such that $\perm{\aut}{\MOD}{\pp}_{h(x)}=\perm{\aut}{\MOD}{\pp}_{a}$.
As $|S'|\geq 2$ by \eqref{allnotid}, $s\in \{s_0,\dots,s_{\pp-1}\}$ must hold, and so
$\{s_0,\dots,s_{\pp-1}\}\subseteq S'$ follows by \eqref{sclosed}.  As there is
$y\in\{u,v\}^\ast$ with $\perm{\aut}{\MOD}{\pp}_{h(y)}=\perm{\aut}{\MOD}{\pp}_{\natural}$, $s_\pp\in S'$ also follows by \eqref{sclosed}.
Finally, as $\{u,v\}^\ast\not\subseteq\natural^\ast$, there is $x\in\{u,v\}^\ast$ of the form
$\natural^na_1wa_2w'$, for some $n<\omega$, $w\in\Aiabc$ and $w'\in\Aabc^\ast$.
As $S'=S$, $\perm{\aut}{\MOD}{\pp}_{x}(s_i)\in S$ for every $i\leq p$, and so
$w\in\bigcap_{i=0}^{\pp} \L(\A_i)$.
\end{proof}


Theorem \ref{DFAhard} clearly follows from Lemmas~\ref{l:DFAs} and \ref{l:empty}.
\end{proof}


\section{Deciding $\lang$-definability of 2NFAs in \PSpace}\label{sec:2nfa}

Using the criterion Theorem~\ref{DFAcrit}~$(i)$, Stern~\cite{DBLP:journals/iandc/Stern85} showed that
deciding whether the language of any given DFA is $\FO(<)$-definable can be done in \PSpace.
In this section, we also use the criteria of Theorem~\ref{DFAcrit} to provide
\PSpace-algorithms deciding whether the language of any given 2NFA is $\lang$-definable, 
whenever $\lang \in \{ \FO(<), \FO(<,\equiv), \FO(<,\MOD)\}$.
Let $\A = (Q, \Sigma, \delta, Q_0, F)$ be a 2NFA. Following~\cite{carton_et_al:LIPIcs:2015:5413}, 
we first construct a(n exponential size) DFA $\A'$ such that $\L(\A)=\L(\A')$. To this end, for any $w \in \Sigma^+$, we introduce four binary relations $\mathsf{b}_{lr}(w)$, $\mathsf{b}_{rl}(w)$, $\mathsf{b}_{rr}(w)$, and $\mathsf{b}_{ll}(w)$ on $Q$ describing the \emph{left-to-right}, \emph{right-to-left}, \emph{right-to-right}, and \emph{left-to-left behaviour of} $\A$ \emph{on} $w$. Namely,
\begin{itemize}
\item $(q,q') \in \mathsf{b}_{lr}(w)$ if there is a run of $\A$ on $w$ from $(q, 0)$ to $(q', |w|)$;

\item $(q,q') \in \mathsf{b}_{rr}(w)$ if there is a run of $\A$ on $w$ from $(q, |w|-1)$ to $(q', |w|)$;

\item $(q,q') \in \mathsf{b}_{rl}(w)$ if, for some $a \in \Sigma$, there is a run on $aw$ from $(q, |aw|-1)$ to $(q', 0)$ such that no $(q'',0)$ occurs in it before $(q', 0)$;

\item $(q,q') \in \mathsf{b}_{ll}(w)$ if, for some $a \in \Sigma$, there is a run on $aw$ from $(q, 1)$ to $(q', 0)$ such that no $(q'',0)$ occurs in it before $(q', 0)$.
\end{itemize}
For $w = \varepsilon$ (the empty word), we define the $\mathsf{b}_{ij}(w)$ as the identity relation on $Q$.
Let $\mathsf{b} = (\mathsf{b}_{lr}, \mathsf{b}_{rl}, \mathsf{b}_{rr}, \mathsf{b}_{ll})$, where the $\mathsf{b}_{ij}$ are the behaviours of $\A$ on some $w \in \Sigma^*$, in which case we can also write $\mathsf{b}(w)$, and let $\mathsf{b}' = \mathsf{b}(w')$, for some $w' \in \Sigma^*$. We define the composition $\mathsf{b} \cdot \mathsf{b}' = \mathsf{b}''$ with components $\mathsf{b}_{ij}''$ as follows. Let $X$ and $Y$ be the transitive closure of $\mathsf{b}_{ll}' \circ \mathsf{b}_{rr}$ and $\mathsf{b}_{rr} \circ \mathsf{b}_{ll}'$, respectively.
Then we set:
\begin{align*}
&  \mathsf{b}_{lr}'' =  \mathsf{b}_{lr} \circ \mathsf{b}_{lr}' \cup  \mathsf{b}_{lr} \circ X \circ \mathsf{b}_{lr}',\qquad
\mathsf{b}_{rl}'' =  \mathsf{b}_{rl}' \circ \mathsf{b}_{rl} \cup \mathsf{b}_{rl}' \circ Y \circ \mathsf{b}_{rl}, \\
&  \mathsf{b}_{rr}'' =  \mathsf{b}_{rr}' \cup \mathsf{b}_{rl}' \circ Y \circ \mathsf{b}_{rr} \circ \mathsf{b}_{lr}',\qquad
 \mathsf{b}_{ll}'' = \mathsf{b}_{ll} \cup \mathsf{b}_{lr} \circ X \circ \mathsf{b}_{ll}' \circ \mathsf{b}_{rl}.
\end{align*}
One can check that $\mathsf{b}'' = \mathsf{b}(ww')$.
Define a DFA $\A' = (Q', \Sigma, \delta', q_0', F')$ by taking
\begin{align*}
& Q' = \bigl\{ (B_{lr}, B_{rr}) \mid B_{lr} \subseteq Q_0 \times Q, \ B_{rr} \subseteq Q \times Q \bigr\},\ \
q_0' = \bigl(\bigl\{(q,q) \mid q \in Q_0\bigr\}, \emptyset\bigr),\\
& F' = \bigl\{(B_{lr}, B_{rr}) \mid (q_0, q) \in B_{lr}, \text{ for some $q_0 \in Q_0$ and $q\in F$} \bigr\},\\
& \delta'_a\bigl((B_{lr}, B_{rr})\bigr) =(B_{lr}', B_{rr}'), \text{ with}\ B_{lr}' = B_{lr} \circ X(a) \circ \mathsf{b}_{lr}(a),\\
& \hspace*{5.2cm} B_{rr}' = B_{rr} \cup \mathsf{b}_{rl}(a) \circ Y(a) \circ \mathsf{b}_{lr}(a),
\end{align*}
where $X(a)$ and $Y(a)$ are the reflexive and transitive closures of $\mathsf{b}_{ll}(a) \circ B_{rr}$ and $B_{rr} \circ \mathsf{b}_{ll}(a)$, respectively.
It is not hard to see that, for any $w\in\Sigma^\ast$,
\begin{align}
\nonumber
& \delta'_w\bigl((B_{lr}, B_{rr})\bigr) =(B_{lr}', B_{rr}')\ \text{iff}\ B_{lr}' = B_{lr} \circ X(w) \circ \mathsf{b}_{lr}(w) \text{ and}\\
\label{aprime-reach}
& \hspace*{4.7cm} B_{rr}' = B_{rr} \cup \mathsf{b}_{rl}(w) \circ Y(w) \circ \mathsf{b}_{lr}(w),
\end{align}
where $X(w)$ and $Y(w)$ are the reflexive and transitive closures of $\mathsf{b}_{ll}(w) \circ B_{rr}$ and $B_{rr} \circ \mathsf{b}_{ll}(w)$, respectively.
Also, one can show in a way similar to~\cite{5392614,10.1016/0020-0190(89)90205-6} that
\begin{equation}\label{eq:2nfatodfa}
\L(\A) = \L(\A').
\end{equation}

Next, we show that, even if the size of $\A'$ is exponential in $\A$, we can still use Theorem~\ref{DFAcrit} to decide $\lang$-definability of $\L(\A)$ in \PSpace:

\begin{theorem}\label{thm:2NFA}
For $\lang \in \{ \FO(<), \FO(<,\equiv), \FO(<,\MOD)\}$,
deciding $\lang$-definability of $\L(\A)$, for any 2NFA $\A$, is in \PSpace.
\end{theorem}
\begin{proof}
Let $\A'$ be the DFA defined above for the given 2NFA $\A$.
By Theorem~\ref{DFAcrit} $(i)$ and~\eqref{eq:2nfatodfa}, $\L(\A)$ is not $\FO(<)$-definable iff there exist a word $u\in\Sigma^\ast$, a reachable state $q \in Q'$, and a number $k\leq|Q'|$ such that $q\not\simm \delta'_u(q)$ and $q= \delta'_{u^k}(q)$. We guess the required $k$ in binary, $q$ and a quadruple $\mathsf{b}(u)$ of binary relations  on $Q$. Clearly, they all can be stored in polynomial space in $|\A|$. To check that our guesses are correct, we  first check that $\mathsf{b}(u)$ indeed corresponds to some $u \in \Sigma^\ast$. This is done by guessing a sequence $\mathsf{b}_0, \dots, \mathsf{b}_n$ of distinct quadruples of binary relations on $Q$ such that $\mathsf{b}_0 = \mathsf{b}(u_0)$ and $\mathsf{b}_{i+1} = \mathsf{b}_{i} \cdot \mathsf{b}(u_{i+1})$, for some $u_0, \dots, u_n \in \Sigma$. (Any sequence with a subsequence starting after $\mathsf{b}_i$ and ending with $\mathsf{b}_{i+m}$, for some $i$ and $m$ such that $\mathsf{b}_i = \mathsf{b}_{i+m}$, is equivalent, in the context of this proof, to the sequence with such a subsequence removed.) Thus, we can assume that $n \leq 2^{O(|Q|)}$, and so $n$ can be guessed in binary and stored in \PSpace{}. So, the stage of our algorithm checking that $\mathsf{b}(u)$ corresponds to some $u \in \Sigma^*$ makes $n$ iterations and continues to the next stage if $\mathsf{b}_n = \mathsf{b}(u)$ or terminates with an answer \no{} otherwise. Now, using $\mathsf{b}(u)$, we compute $\mathsf{b}(u^k)$ by means of a sequence $\mathsf{b}_0, \dots, \mathsf{b}_k$, where $\mathsf{b}_0 = \mathsf{b}(u)$ and $\mathsf{b}_{i+1} = \mathsf{b}_{i} \cdot \mathsf{b}(u)$. With $\mathsf{b}(u)$ ($\mathsf{b}(u^k)$), we compute $\delta'_{u}(q)$ (respectively, $\delta'_{u^k}(q)$) in \PSpace{} using \eqref{aprime-reach}. If $\delta'_{u^k}(q) \neq q$, the algorithm terminates with an answer \no{}. Otherwise, in the final stage of the algorithm, we check that $\delta'_{u}(q) \not \sim q$. This is done by guessing $v \in \Sigma^*$ such that $\delta'_v(q) = q_1$, $\delta'_v\bigl(\delta'_{u}(q)\bigr) = q_2$, and $q_1 \in F'$ iff $q_1 \not \in F'$. We guess such a $v$ (if exists) in the form of $\mathsf{b}(v)$ using an algorithm analogous to that for guessing $u$ above.

By Theorem~\ref{DFAcrit}~$(ii)$ and~\eqref{eq:2nfatodfa}, $\L(\A)$ is not $\FO(<,\equiv)$-definable iff there
there exist words $u,v\in\Sigma^\ast$, a reachable state $q \in Q'$, and a number $k\leq|Q'|$ such that
$q\not\simm\delta'_u(q)$, $q=\delta'_{u^k}(q)$, $|v|=|u|$, and $\delta'_{u^i}(q)=\delta'_{u^iv}(q)$, for all $i<k$.
We outline how to modify the algorithm for $\FO(<)$ above to check $\FO(<,\equiv)$-definability. First, we need to guess and check $v$ in the form of $\mathsf{b}(v)$ in parallel with guessing and checking $u$ in the form of $\mathsf{b}(u)$, making sure that $|v| = |u|$. For that, we guess a sequence of distinct pairs $(\mathsf{b}_0, \mathsf{b}_0'), \dots, (\mathsf{b}_n, \mathsf{b}_n')$ such that the $\mathsf{b}_i$ are as above, $\mathsf{b}_0' = \mathsf{b}(v_0)$ and $\mathsf{b}_{i+1}' = \mathsf{b}_{i}' \cdot \mathsf{b}(v_{i+1})$, for some $v_0, \dots, v_n \in \Sigma$. (Any such sequence with a subsequence starting after $(\mathsf{b}_i, \mathsf{b}_i')$ and ending with $(\mathsf{b}_{i+m}, \mathsf{b}_{i+m}')$, for some $i$ and $m$ such that $(\mathsf{b}_i, \mathsf{b}_i') = (\mathsf{b}_{i+m}, \mathsf{b}_{i+m}')$, is equivalent to the sequence with that  subsequence removed.) So $n \leq 2^{O(|Q|)}$.
For each $i<k$,
we can then compute $\delta'_{u^i}(q)$ and $\delta'_{u^iv}(q)$, using \eqref{aprime-reach}, and check whether
whether they are equal.

Finally,
by Theorem~\ref{DFAcrit}~$(iii)$ and~\eqref{eq:2nfatodfa}, $\L(\A)$ is not $\FO(<,\MOD)$-definable iff
 there exist $u,v\in\Sigma^\ast$, a reachable
state $q\in Q'$ and $k,l\leq |Q'|$ such that $k$ is an odd prime, $l>1$ and coprime to both $2$ and $k$,
$q\not\simm\delta'_u(q)$, $q\not\simm\delta'_v(q)$, $q\not\simm\delta'_{uv}(q)$, and
$\delta'_{x}(q)\simm\delta'_{xu^{2}}(q)\simm\delta'_{xv^{k}}(q)\simm\delta'_{x(uv)^{l}}(q)$, for all $x\in\{u,v\}^\ast$.
We start by guessing $u,v \in \Sigma^*$ in the form of $\mathsf{b}(u)$ and $\mathsf{b}(u)$, respectively.
Also, we guess $k$ and $l$ in binary and check that $k$ is an odd prime and $l$ is coprime to both $2$ and $k$.
By \eqref{aprime-reach}, $\delta'_x$ is determined by $\mathsf{b}(x)$, for any $x\in\{u,v\}^\ast$.
Thus, we can proceed as follows to verify that $u$, $v$, $k$ and $l$ are as required. We perform the following steps, for \emph{each} quadruple $\mathsf{b}$ of binary relations on $Q$. First, we check whether $\mathsf{b} = \mathsf{b}(x)$, for some $x \in \{u,v\}^\ast$ (we discuss the algorithm for this below). If this is not the case, we construct the \emph{next} quadruple $\mathsf{b}'$ and process it as this $\mathsf{b}$. If it is the case,
we compute all the states $\delta'_{x}(q)$, $\delta'_{xu^{2}}(q)$, $\delta'_{xv^{k}}(q)$, $\delta'_{x(uv)^{l}}(q)$, $\delta'_{u}(q)$, $\delta'_{v}(q)$, $\delta'_{uv}(q)$,
and check their required (non)equivalences w.r.t.\ $\sim$,
using the same method as for checking $\delta'_{u}(q) \not\sim q$ above.
 If they do not hold as required, our algorithm terminates with an answer \no{}. Otherwise, we construct the \emph{next} quadruple $\mathsf{b}'$ and process it as this $\mathsf{b}$. When all possible quadruples $\mathsf{b}$ of binary relations of $Q$ have been processed, the algorithm terminates with an answer \yes{}.

Now, to check that a given quadruple $\mathsf{b}$ is equal to $\mathsf{b}(x)$, for some $x \in  \{u,v\}^\ast$, we simply guess a sequence $\mathsf{b}_0, \dots, \mathsf{b}_n$ of quadruples of binary relations on $Q$ such that $\mathsf{b}_0 = \mathsf{b}(w_0)$, $\mathsf{b}_n = \mathsf{b}$ and $\mathsf{b}_{i+1} = \mathsf{b}_{i} \cdot \mathsf{b}(w_{i+1})$, where $w_i \in \{u, v\}$. It follows from the argument above that it is enough to consider $n \leq 2^{O(|Q|)}$.
\end{proof}

\section{Further Research}

The results obtained in this paper have been used for deciding the rewritability type of ontology-mediated queries (OMQs) given in linear temporal logic \LTL~\cite{DBLP:conf/time/21}.
As mentioned in the introduction, \LTL{} OMQs can be simulated by automata. In the worst case, the automata are of exponential size, and deciding FO-rewritability of some OMQs may become \ExpSpace-complete. On the other hand, there are natural and practically important fragments of \LTL{} with automata of special forms whose FO-rewritability can be decided in \PSpace{}, $\Pi^p_2$ or \coNP{}. However, it remains to be seen whether the corresponding algorithms, even in the simplest case of $\FO(<)$-definability, are efficient enough for applications in temporal OBDA. Note that the problems considered in this paper are also relevant to the optimisation problem for recursive SQL queries. 

\medskip
\noindent
{\bf Acknowledgements.} This work was supported by UK EPSRC EP/S032282.



\end{document}